\documentclass[11pt, reqno]{amsart}

\usepackage{amsfonts, amssymb, amscd, amsrefs}
\usepackage{graphicx}
\usepackage{float}
\usepackage{hyperref}
\usepackage{slashed}
\usepackage{fullpage}
\usepackage{tikz}

\usepackage{bbold}
\makeatletter
\def\amsbb{\use@mathgroup \M@U \symAMSb}
\makeatother

\theoremstyle{plain}
\newtheorem{thm}{Theorem}
\newtheorem{lem}[thm]{Lemma}
\theoremstyle{remark}
\newtheorem{remark}[thm]{Remark}

\newcommand{\E}{\amsbb{E}}
\newcommand{\R}{\amsbb{R}}

\newcommand{\Var}{\mathrm{Var}}
\newcommand{\market}{\mathrm{market}}
\newcommand{\NoFill}{\mathrm{no\;fill}}
\newcommand{\Fill}{\mathrm{fill}}
\newcommand{\MarketFill}{\mathrm{market\;fill}}
\newcommand{\LimitFill}{\mathrm{limit\;fill}}
\newcommand{\takerfee}{c_{\mathrm{taker}}}
\newcommand{\makerfee}{c_{\mathrm{maker}}}

\title{The effect of latency on optimal order execution policy}

\author{Chutian Ma$^{*}$}
\address[Chutian Ma]{Department of Mathematics, Johns Hopkins University and Causify.AI Inc.}
\email{cma27@jhu.edu}

\author{Giacinto Paolo Saggese$^{*}$}
\address[Giacinto Paolo Saggese]{Department of Computer Science, University of Maryland, College Park and Causify.AI Inc.}
\email{gsaggese@umd.edu}

\author{Paul Smith$^{*}$}
\address[Paul Smith]{Department of Mathematics, University of North Carolina-Chapel Hill and Causify.AI Inc.}
\email{pasmit@unc.edu}
\thanks{$^{*}$ Authors listed alphabetically.}

\begin{document}

\begin{abstract}
Market participants regularly send bid and ask quotes to exchange-operated
limit order books. This creates an optimization challenge where their
potential profit is determined by their quoted price and how often their
orders are successfully executed. The expected profit from successful
execution at a favorable limit price needs to be balanced against two key
risks: (1) the possibility that orders will remain unfilled, which hinders
the trading agenda and leads to greater price uncertainty, and (2) the
danger that limit orders will be executed as market orders, particularly in
the presence of order submission latency, which in turn results in higher
transaction costs. In this paper, we consider a stochastic optimal control
problem where a risk-averse trader attempts to maximize profit while
balancing risk. The market is modeled using Brownian motion to represent the
price uncertainty. We analyze the relationship between fill probability,
limit price, and order submission latency. We derive closed-form
approximations of these quantities that perform well in the practical regime
of interest. Then, we utilize a mean-variance method where our total reward
function features a risk-tolerance parameter to quantify the combined risk
and profit. We show that the optimal policy is characterized by the
Hamilton–Jacobi–Bellman equation and can be computed efficiently from
backward dynamic programming. 
\end{abstract}

\maketitle

\tableofcontents

\section{Introduction}

When trading financial assets on exchanges using limit order books, traders
typically use a combination of limit orders and market orders to execute their
trades. A market order is an instruction to buy or sell an asset immediately at
the best available price, ensuring quick execution. A limit order, on the other
hand, allows traders to set a specific price at which they are willing to buy or
sell an asset. The limit order is not executed unless a buyer/seller is willing
to pay the price specified in the limit order. 

Limit orders would be preferable to a trader, if only execution were guaranteed,
because they provide critical price control and fee advantages in high-volume
trading environments. By specifying exact execution prices, these market
participants can precisely manage their entry and exit points. In the case of
some participants, such as market makers, this precise control is necessary to
capture the bid-ask spread, which may be the trader's primary source of profit.
Limit orders also typically qualify for maker rebates rather than incurring
taker fees, substantially improving economics when executing thousands of trades
daily. This is supported by the empirical study in \cite{harris1996market}.

If a limit order is filled, its expected return is its price improvement
compared to a benchmark. Since our work is focused on order execution quality,
we do not explicitly consider the role of a filled order in the context of a
portfolio management strategy, as this role can vary from strategy to strategy
and depend upon numerous factors. If the limit order is not filled, then there
is no return to consider. Thus, the investor must assess the probability that
the limit order gets executed within a certain time window, and balance it with
the expected return. For example, setting a high ask price for a sell order
yields better return given a fill, but a low probability of execution.  

From a mathematical point of view, the order execution problem can be viewed as
the choice of optimal bid/ask quotes and execution window. The constraints on
the optimization arise from many factors, including risk tolerance of the trader
and market impact. Optimal execution strategies which address different concerns
have been studied. For example, \cite{almgren2001optimal} studies the optimal
execution strategy taking into account market impact and price volatility. The
work in \cite{obizhaeva2013optimal} advances this direction by further using a
liquidity recovery model to measure the market impact while studying execution
strategy. See also \cite{gatheral2011optimal}, \cite{gueant2013dealing},
\cite{predoiu2011optimal} for references. Our work in this paper mainly
addresses the impact of the transaction costs caused by latency and the
uncertainty of the price on determining the optimal execution strategy.

The stochastic nature of the problem comes from the uncertainty of the price. We
will use Brownian motion to model the price movement in a short time window. As
is well known, the usage of Brownian motion in modeling the price of financial
assets can be traced back to Louis Bachelier's work \cite{bachelier1900theorie}
in 1900. This approach gained mainstream acceptance following Black and Scholes'
and Merton's work on the option pricing model, see \cite{black1973pricing},
\cite{merton1971theory}. For research in this direction, see also
\cite{garleanu2016dynamic}, \cite{heston1993closed},
\cite{merton1973intertemporal}, \cite{hu1998optimal}. Price has also been
modeled using other stochastic processes, for example the compound Poisson
process. We refer interested readers to \cite{cohen1981transaction},
\cite{cho2000probability} for references. In this paper, we assume the price
follows a Brownian motion during a short time horizon. Various important
features of trading, including probabilities of orders being filled,
probabilities of a market fill given a fill, and expected reverse price movement
given no fill, are calculated using this model. 

As we mentioned earlier, one of the factors we consider that needs to be
balanced with profit is the effect of latency (or lag). Latency refers to the
delay between the time when the investor sends an order and the actual time when
the order reaches the order book. In modern markets, latency can range from
microseconds to hundreds of milliseconds. Longer latency adds to the transaction
costs in a variety of ways:
\begin{enumerate}
    \item Taker/maker fee.
    Investors with higher latency take a greater risk, as the price is more
    likely to drift away from the price at the time of decision. This may result
    in non-execution of the order or may result in their limit orders becoming
    marketable by the time it arrives in the order book. A market order incurs a
    taker fee while an executed limit order may compensate the trader with a
    maker fee. For large-volume traders and market makers, the fees can
    accumulate quickly if their limit orders are not placed properly.
    \item Queue position.
    In most modern markets, orders with the same price are executed according to
    their arrival time. The traders with lower latency gain an advantage, as
    their orders have higher priority, and in turn suffer less from potential
    negative price movement.
\end{enumerate}
We will address the first effect, i.e., risk of non-execution and potential
maker/taker fees, using our model. We examine the probability of ``marketable
limit orders" as a function of the latency, which is then incorporated into the
reward function as a potential loss (negative profit). We will assume the
latency is a constant in this paper. One may also model the latency itself as a
random variable. We refer the reader to Cartea and Sánchez-Betancourt's work in
\cite{cartea2023optimal} where they incorporate random delays and model the
limit order book dynamics with both fundamental price movements and short-lived
flickering quotes. They demonstrate that optimal strategies differ significantly
between patient and impatient traders, with the former using limit orders
primarily for speculative price improvements. Their empirical analysis of FX
market data shows that latency-optimal strategies outperform conventional
approaches by a scale comparable to the transaction costs, further justifying
the necessity of considering latency in the order execution framework. For more
work in this direction, we refer the readers to \cite{Le21},
\cite{cartea2021latency}, \cite{cartea2021shadow} and \cite{gao2020optimal}. 

With our model for price in hand, the stage is set for the stochastic optimal
control problem, which aims to optimize the reward balanced with risk. We assume
that the trader is risk-averse and is seeking to execute orders within a given
time window. There are two major choices to be made. One choice is the limit
price. A high price increases the risk of non-execution and a low price
increases the risk of execution as a market order. The other choice is the
execution window. In addition to taking into account hard constraints like a
trading agenda, traders consider how price uncertainty increases as a function
of the trading horizon, and so typically prefer prompt order execution. We
introduce elements of mean-variance analysis in our formulation of the
stochastic control problem so as to balance profits and the risk discussed
above. We define our reward function using the exponential utility function,
which has been developed for risk-sensitive controls, see
\cite{nagai1996bellman}, \cite{marcus1997risk}. The exponential utility function
introduces an extra parameter $\lambda$, which measures the risk tolerance of
the trader. A larger $\lambda$ indicates a greater level of risk aversion. We
then study the risk-adjusted optimal policy obtained from the Bellman equation.
We would like to mention the work in \cite{Le21}, where the optimization of the
order execution reward with latency effects is considered. The reward functions
considered in \cite{Le21} specifically feature the market impact caused by
walking the order book. Our work emphasizes (1) the potential price slippage and
negative price movement caused by latency if an order is not executed and (2)
optimal execution for small-volume orders consuming not more than the liquidity
available at the top of the order book.

The outline of the paper is as follows. In Section \ref{sec: bm}, we review the
basics of the Brownian motion price model and describe our order execution
model. We provide estimates for the probability of fills as both limit and
market orders, as well as estimates of the non-fill probability. In Section
\ref{sec: order split}, we aggregate the probabilities we calculated in the
previous section and formulate the order execution problem as a stochastic
optimal control problem. We study the optimal policy, which is affected by a
number of factors including instrument volatility, bid-ask spread, latency, and
maker/taker fees. In Section \ref{sec: mc sim}, we present results obtained by
Monte Carlo simulations. 

We use arrival price as our price benchmark.
Notably, we exclude certain real-world effects, such as queue position and the
possibility of partial fills (though our analysis may be extended to handle the
latter).

\section{The Brownian motion price model}
\label{sec: bm}
The basic model that we adopt for modeling price is that of a standard Brownian
motion (BM) on the unit interval. Varying instrument volatilities and limit
order durations may be rescaled so that this model applies.

We choose to adopt the standard BM as our price model rather than geometric
Brownian motion primarily because our main interest is in short time scales
(e.g., see the discussion in \cite{BoBoDoGo18}[\S 2.1.1]). We note that a
geometric Brownian motion model was similarly considered in \cite{LoMaZh02}. Our
conclusions regarding the applicability of the standard BM model to our domain
differ from those of \cite{LoMaZh02}, but this is perhaps due to the orders of
magnitude difference in the time scales considered.

We note that, for our initial analysis, we do not consider a BM with drift. We
again justify this based upon our primary interest in short time scales. One may
extend our analysis to BM with drift at the cost of increased analytical
complexity.

\subsection{Formalism}

Without loss of generality, we restrict ourselves to ``sell" orders so as to
simplify the exposition. Let $B_t$ be a standard Brownian motion on the unit
interval $t \in [0, 1]$. We take $B_t$ to represent the top-of-book ``bid" price
(translated to zero at time zero). Note that these are ``opposite-side
quote-relative coordinates", using the terminology of \cite{BoBoDoGo18}[\S
3.1.6]. Each limit order decision involves choosing a value $y \in \R$ (where
the case $y \leq 0$ corresponds to a marketable limit order). Note that $y$ is
expressed in units of volatility, as $B_t$ is a standard BM.

The basic mechanics of limit order placement and execution in our model are as
follows:
\begin{itemize}
  \item At time $t_0 = 0$, we observe the order book (bid taken to be zero) and
        choose a limit price $y \in \R$ at which to place a sell limit order.
  \item At some time $t_1 > t_0$, the limit order hits the limit order book.
  \item If, at time $t_1$, the current bid price, as represented by $B_{t_1}$,
      is greater than or equal to $y$, then we consider the order to be
      marketable and fully executed at time $t_1$.
  \item If $B_{t_1} < y$, then the order rests on the book as a limit order. If,
      for some $t_1 < t \leq 1$, we have $B_t = y$, then we consider the order
      fully executed at time $t$.
  \item If we arrive at time $t = 1$ with the order unexecuted, then the
        order expires and so is unfilled.
\end{itemize}

\subsection{Single limit order execution: hitting times}

Given our formal price model, it is natural to study the distribution of
hitting times of a BM as a function of $y > 0$.
Let
\[
    T_y = \inf \{t \geq 0: B_t = y\}.
\]
Then $T_y$ is Levy distributed with scale parameter $y^2$, which, for $y > 0$,
has probability density function (pdf)
\[
    g_y(t) = \frac{y}{\sqrt{2 \pi t^3}} \exp \left( - \frac{y^2}{2 t} \right) \mathbb{1}_{\{t > 0\}}
\]
and cumulative density function (cdf) given by
\[
    P(T_y \leq t) = 2 \Phi \left( - \frac{y}{\sqrt{t}} \right),
\]
where $\Phi$ is the standard normal cdf, defined via
\[
  \Phi(x) = \int_{-\infty}^x f(t) dt
\]
with
\[
    f(t) = \frac{1}{\sqrt{2 \pi}} e^{-\frac{1}{2} t^2}.
\]




\subsection{Risk of execution as a market order}

Suppose there is a predictable consistent latency associated with order
submission. That is, the time it takes to observe the price at $t = 0$, make an
order submission decision, submit the order, and for the order to arrive in the
limit order book, is a given quantity, say $0 < \ell < 1$. The order will only
execute if the price reaches the limit price threshold between time $t = \ell$
and $t = 1$, which is within the execution window and after the order reaches
the limit order book. The order will be marketable if the price at $t = \ell$
already exceeds the threshold. In our framework, we take this to be the
probability that a limit order will be marketable by the time it hits the limit
order book. More formally, let $y \geq 0$ denote the limit price of our order.
Denote by $P_y(\text{market fill})$ the probability that the limit is marketable
by the time it reaches the order book. We have
\begin{align}\label{eq prob market fill}
    P_y(\text{market fill}) = P(B_\ell \geq y) = \Phi(-\frac{y}{\sqrt{\ell}}).
\end{align}
The lower (smaller) the latency, the less likely a market fill will occur. In
fact, the limit price needs to be at least comparable to $\sqrt{\ell}$ in order
for the market fill probability to be low.

The following plot displays these market fill probabilities for a range of
delays (expressed as fractions of the execution window) over a range of limit
order price levels (normalized in units of volatility) from 0 to 1.5.

\begin{figure}[H]
   \centering
   \includegraphics[width=\linewidth]{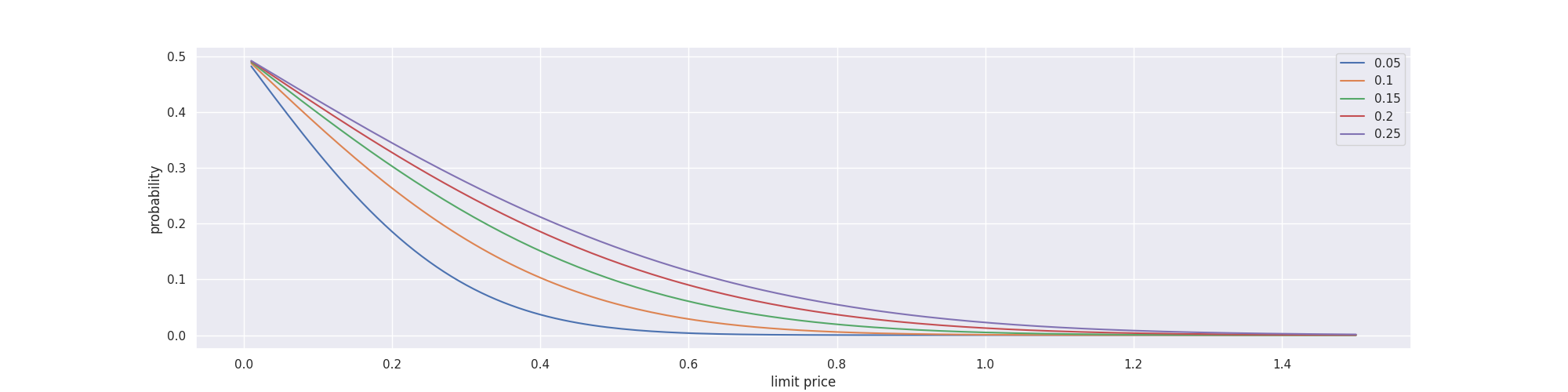}
   \caption{Probability of a market fill}
\end{figure}

The probability of a limit fill can be computed by conditioning on the value of
$B_\ell$. Suppose $B_\ell = z$ where $z < y$. Then the order will be executed as
a limit order if the price moves up by at least $y - z$ in the remaining $1 -
\ell$ time window:
\begin{equation}\label{eq prob limit fill int}
    \begin{aligned}
        P_y(\LimitFill) & = \int_{-\infty}^{y} P(\sup_{\ell < t < 1} B_t \geq y | B_\ell = z) P(B_\ell \in dz) \\
        & = \int_{-\infty}^{y} P(T_{y-z} \leq 1-\ell) P(B_\ell \in dz) \\
        & = \int_{-\infty}^{y} 2\Phi(-\frac{y-z}{\sqrt{1-\ell}}) \frac{1}{\sqrt{\ell}} f(\frac{z}{\sqrt{\ell}}) dz.
    \end{aligned}
\end{equation}
In practice, one may prefer a more efficient way to evaluate the probability than
through integration. In this paper, we will approximate $P_y(\text{limit fill})$ using 
\begin{equation}\label{eq prob limit fill approx}
    P_y(\LimitFill) \approx 2 \Phi(-y) ( 1 - \Phi(-
    \frac{y}{\sqrt{\ell}})) + \frac{2}{\sqrt{2\pi}} \sqrt{\ell}
    f(-\frac{y}{\sqrt{\ell}}).
\end{equation}
\begin{remark}
    The derivation of the approximation may be found in the Appendix.
\end{remark}
\begin{remark}
    Heuristically, one may view the events $\{T_y < 1\}$ and $\{B_\ell < y\}$ as
    ``almost" independent for a small delay value $\ell$. The first term in
    \eqref{eq prob limit fill approx} is the product of these two
    ``almost-independent" probabilities. The second term is a lower-order
    correction term.
\end{remark}
The approximation works well for delay $\ell$ in the range $\ell < 0.2$. The
following figure shows the shape of the probability of a limit fill for a range
of delays.
\begin{figure}[H]
    \centering
    \includegraphics[width=\linewidth]{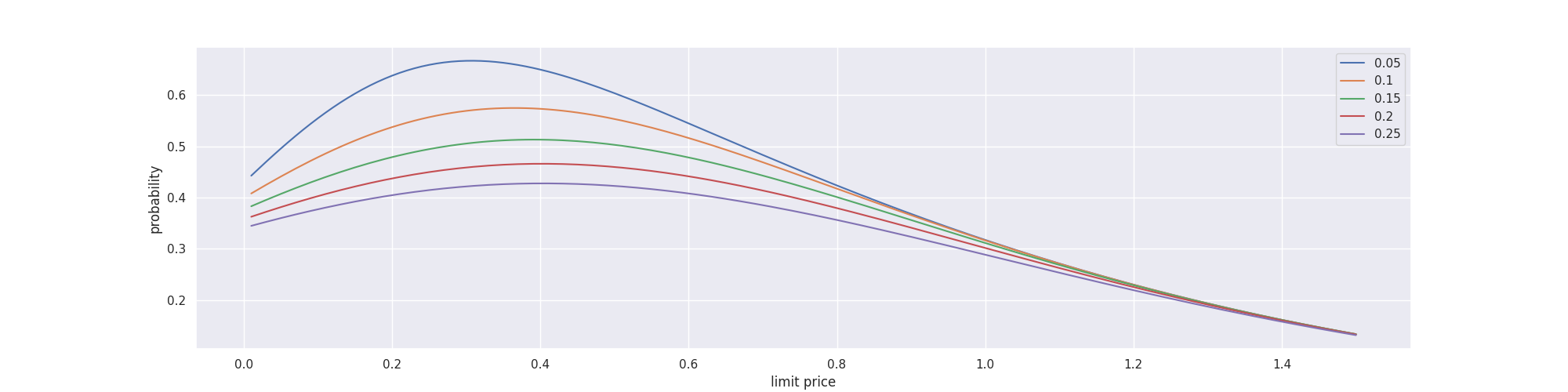}
    \caption{Probability of a limit fill}
\end{figure}
Lastly, the probability of no fill is equal to
\begin{equation}
    P_y(\NoFill) = 1 - P_y(\MarketFill) - P_y(\LimitFill).
\end{equation}

Note that, in the small limit price range, increasing the limit price makes the
order more likely to be executed as a limit order due to their being a smaller
chance of a market fill. Hence, there is a local maximum, and beyond it, the
risk of a non-fill becomes the dominant factor, which causes the limit fill
probability to drop. 

In some analyses, more relevant than the marketable probability is the
conditional probabilty that a filled order is filled as a market order rather
than a limit order. The consideration of this risk is important particularly for
trading that may exhibit a market-making component where one is seeking to earn
rebates by providing liquidity rather than pay fees by taking liquidity. The
following plot shows the expected proportion of market fills among all fills.

\begin{figure}[ht!]
    \centering
    \includegraphics[width=\linewidth]{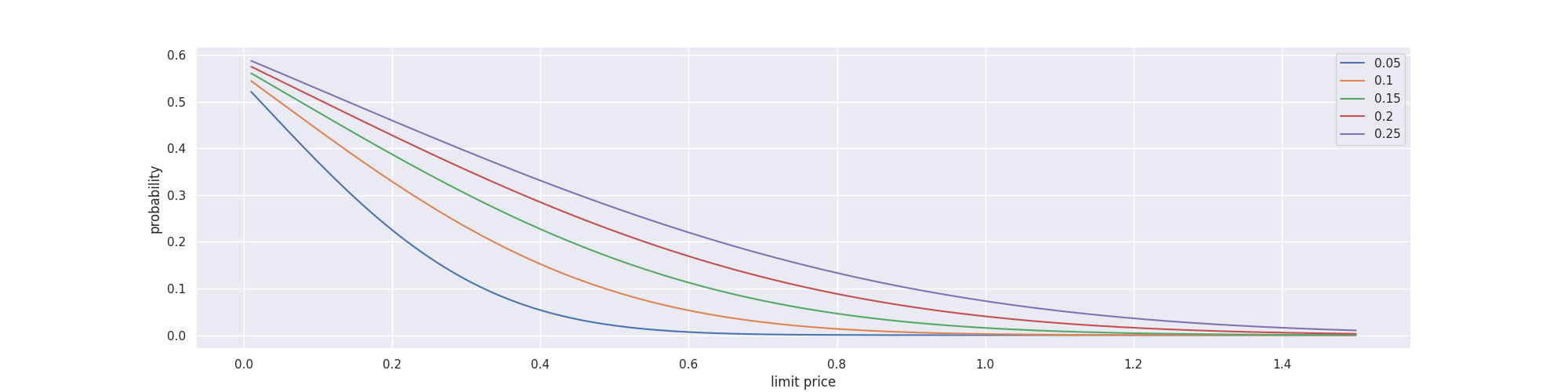}
    \caption{Expected proportion of market fills among all fills}
\end{figure}

\begin{remark}[Effects of scaling]
Suppose now that we change the period of execution, but preserve the standard
scaling of unit volatility over a time interval of length one. If the new time
interval is of length $T$, then we have that the volatility over the length $T$
window is $\sqrt{T}$. To preserve the execution probability, we must rescale the
limit price to some $y^\prime$ and consider the execution at $t^\prime = 1$,
where $t^\prime = \sqrt{T} t$. This implies that we take
\[
  y^\prime = \sqrt{T} y.
\]
For windows of time with $T < 1$, this means placing the limit closer to the
market, and for windows $T > 1$, placing the limit farther away.
\end{remark}



\begin{remark}[Effect of spread]
Orders that are sufficiently aggressive may end up inside the spread. If the
order is closer to the market than the bid-ask midpoint, then execution would be
considered to have entailed spread costs.

Since not paying spread costs is an execution goal, the spread (in addition to
latency) imposes a limit on how aggressively orders may be placed. Note also
that spread costs must be measured with respect to a benchmark. Two obvious
candidates are (1) the spread associated with the arrival price and (2) the
spread immediately before or after execution. Since the arrival price spread is
usable in decision-making processes, it is the benchmark that we adopt in this
work.
\end{remark}

\subsection{Close price given non-execution}
When a limit order is not executed within a window, we may infer that the price
is likely to have moved away from the limit price. This heuristic is made
precise by estimating the expected close given a non-fill.

\subsubsection{Expected close given a non-fill}
To compute the expected close given a non-fill, we circumvent pdf calculations
by using the martingale property.

Let $y$ be the limit order price and let $p = p(y)$ denote the probability of a
fill with limit $y$ during the trading window. Then, by the martingale property,

\begin{align*}
    & \E[B_1 | \LimitFill] = y.
\end{align*}
    
We must have
\[
  \E[B_1] = \E[B_1| \MarketFill] P(\MarketFill) + 
  \E[B_1 | \LimitFill] P(\LimitFill) + \E[B_1 | \NoFill] P(\NoFill),
\]
which yields
\begin{equation}\label{eq expected close given no fill}
    \E(B_1|\NoFill) \approx
    - \frac{-\sqrt{\ell} f(\frac{y}{\sqrt{\ell}}) - y \cdot 2\Phi(-y) (1 -
    \Phi(-\frac{y}{\sqrt{\ell}})) - \frac{2}{\sqrt{2\pi}}\sqrt{\ell}
    f(\frac{y}{\sqrt{\ell}})}{1 - \Phi(-\frac{y}{\sqrt{\ell}}) - 2 \Phi(-y) ( 1
    - \Phi(- \frac{y}{\sqrt{\ell}})) - \frac{2}{\sqrt{2\pi}} \sqrt{\ell}
    f(-\frac{y}{\sqrt{\ell}})}.
\end{equation}
The behavior of the expected close is presented in the figure below.
\begin{figure}[H]
    \centering
    \includegraphics[width=\linewidth]{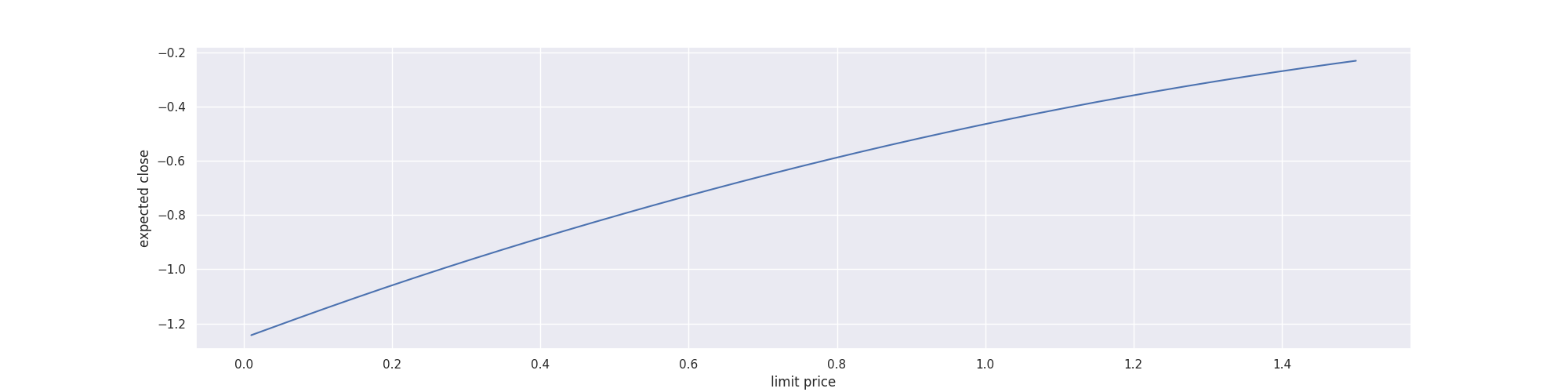}
    \caption{Expected close price given non-execution of a limit order}
\end{figure}

This is in line with the intuition. If the Brownian motion fails to reach a very
small threshold, then the Brownian motion has likely moved toward the other
direction. On the other hand, failing to pass a large threshold has little
effect on our inference.

This intuition can also be viewed from the conditional probability distribution
of the close price given a nonfill. The derivation of the following expression
can be found in \eqref{eq prob close given nonfill} in the Appendix:
\begin{equation}
    p(c|h \leq y) = \frac{1}{\sqrt{2 \pi}} e^{-\frac{1}{2} c^2}
    \frac{1 - e^{-2y(y - c)}}{1 - 2 \Phi(-y)}
    \mathbb{1}_{\{c < y\}}.
\end{equation}
\begin{figure}[H]
    \centering
    \includegraphics[width=\linewidth]{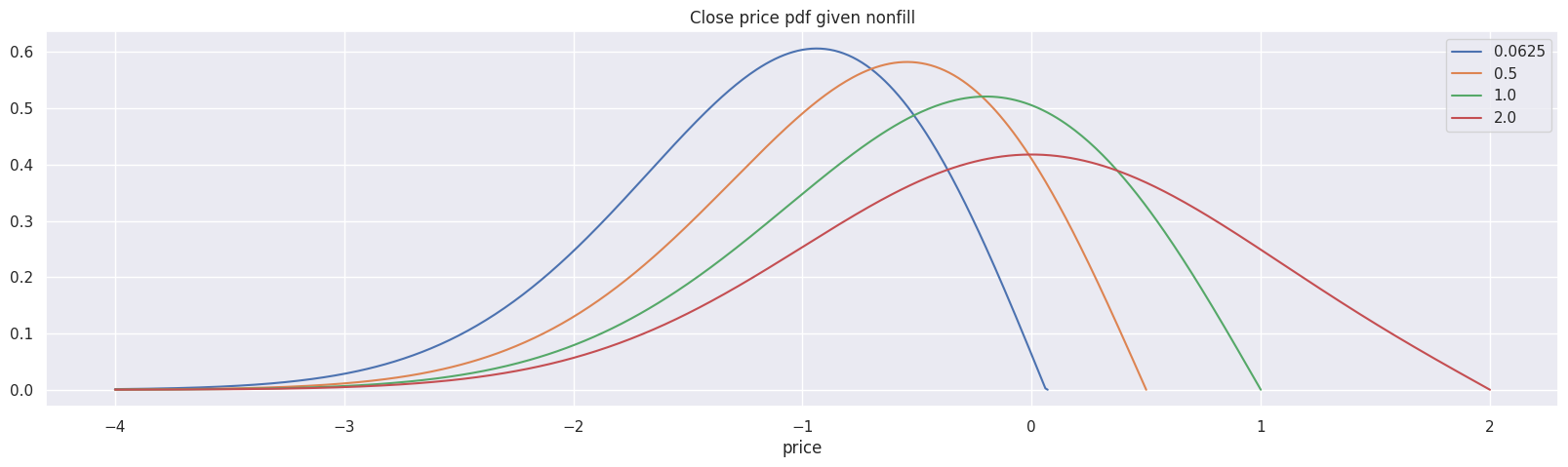}
    \caption{Close price pdf given non-execution of a limit order at
    various limit prices}
\end{figure}
For large values of $y$, this distribution is nearly normal. For values close
to zero, the left-skew becomes more pronounced. In fact, by expanding the
numerators and denominators of the distribution in a Taylor series about the
origin in $y$ and letting $y \to 0^+$, we obtain by L'Hopital's rule
the pointwise limit
\begin{equation} \label{c|h=0_pdf}
    p(c|h=0) = -c e^{-\frac{1}{2} c^2} \mathbb{1}_{\{c < 0\}},
\end{equation}
which after-the-fact can be upgraded to almost sure convergence.
Note that, up to a sign, this is the Rayleigh distribution, which is the
time $t = 1$ distribution of a Brownian meander.

Therefore, in some sense, we may interpret the limit $y$ in the
conditional pdf $p(c|h \leq y)$ as a parameter that controls the level of
interpolation between a Gaussian distribution ($y \to \infty$) on
the one hand, and the Rayleigh distribution ($y \to 0^+$) on the other.




Similarly, we compute the expected close given a fill:
\[
    \E(c | h > y) = y.
\]
The probability distribution of the close price given a fill can be derived from
the joint distribution of close and high as found in the Appendix. 
\begin{figure}[ht!]
    \centering
    \includegraphics[width=\linewidth]{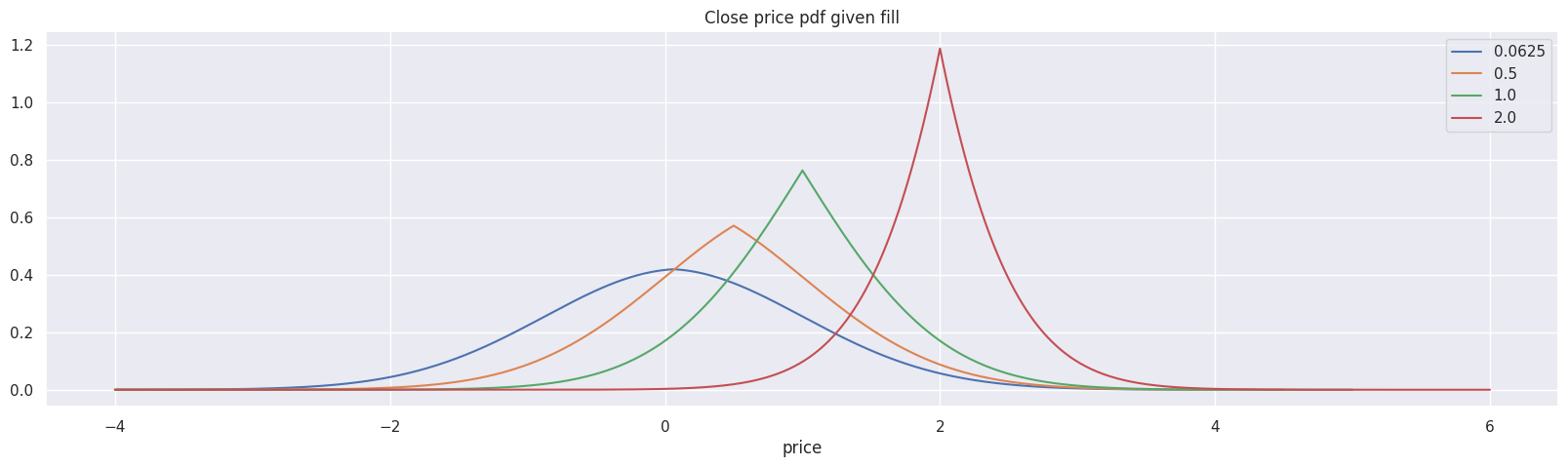}
    \caption{Close price pdf given a fill at various limit prices}
\end{figure}

\section{Basic order splitting}
\label{sec: order split}
In this section, we study the strategy to optimally execute an order. Our basic
strategy is that we place the limit order and reprice it regularly to remain
relevant to the market until the order is finally filled. The effectiveness of
the strategy depends upon the limit price we set and the trading horizon. 

We start by investigating a simple strategy where the limit price offset is
static and the trading horizon is infinite. Then, in the general case, we
formulate the problem as a Markov decision process (MDP). Our model allows us to
quantify and balance the profit, risk, and trading speed at the same time.

\subsection{Naive order splitting}

In the absence of delay, the execution time of a limit order is Levy
distributed. It is notable that the Levy distribution is a stable distribution:
if
$U, U_1, U_2$ are iid standard Levy and $b_1, b_2 > 0$, then
\[
  b_1 U_1 + b_2 U_2 \sim (b_1^{1/2} + b_2^{1/2})^2 U.
\]

Suppose we adopt a naive policy of equally splitting an order of size 1 into $n$
intervals of equal length. Volatility scales as $n^{-1/2}$, and so to keep
distributions comparable, we assume limits scale as $y/\sqrt{n}$. So, each
standard Levy distribution $U_i$ is multiplied by a scaled limit $y^2 / n$ and a
proportion of order $1 / n$. To evaluate the executed quantity within the unit
interval, we consider at time 1 the following random variable:
\[
  \sum_{i = 1}^n \frac{y^2}{n} \cdot \frac{1}{n} U_i \sim y^2 U.
\]
One way to interpret this is that evenly splitting orders and rescaling limits
according to volatility will not change the expected quantity executed, but will
reduce the variance at the cost of the expected profit.

Working at the level of expectations (rather than distributions), we may arrive
at the same conclusion. Since each of the $n$ orders is filled independently and
with equal probability $p$, the number of fills is represented by a binomially
distributed random variable $X \sim B(n, p)$. The expected fill rate is
\[
  \E \frac{1}{n} X  = \frac{1}{n} \cdot np = p
\]
and the dispersion can be calculated from
\[
  \Var(X) = \Var(\frac{1}{n} X) = \frac{1}{n^2} \cdot n p (1 - p)
  = \frac{p (1 - p)}{n}.
\]
With this approach, we see that, by splitting a parent order into $n$ child
orders, we may reduce the variance of the outcome by $1 / n$. In other words, we
achieve a tighter concentration around the expected fill rate.

There is a cost to increasing $n$, summarized by the following points:
\begin{itemize}
  \item Because vol per subinterval scales like $n^{-1/2}$, in real terms limit
      orders must be placed closer to the market (greater risk of market
      execution).
  \item The total overhead time (communication with exchange to submit, confirm,
      and cancel) scales linearly with $n$.
\end{itemize}

As a final note, we observe that equally subdividing a parent order into child
orders will generally lead to significant underfills.

\subsection{Reprice until filled}

Suppose we follow a simple strategy of splitting a parent order into at most $n$
child orders, each with the full amount of unexecuted quantity. For simplicity,
we assume that, if an order is filled, then it is fully filled. So, our strategy
is to submit repriced child orders at regular intervals until the first fill or
until we have placed $n$ orders.

Using our model and assumption on the delay introduced by repricing, we may
estimate the probability of a parent order fill and the probability of a market
order fill.

\subsubsection{Static repricing policy}

Suppose that we choose a fixed (volatility-scaled price offset) $y > 0$ at which
to price each limit order. Then the probability of a fill at interval $k$
follows a geometric distribution, and the probability of a fill within $n$
orders is given by the cdf evaluated at $n$:
\[
  P(\Fill) = 1 - (1 - 2 \Phi(-y))^n.
\]
To illustrate how this probability and the market-fill probability
change as a function of the limit price, number of intervals, and
interval latency percentage, we provide Table \ref{table:static}.

\begin{table}
\begin{tabular}{lllll}
  $y$  & $n$ & $P(\Fill)$ & $\% \; delay$ & $P(\market|\Fill)$ \\
  0.5  & 5   & 0.992      & 10            & 0.184              \\
  0.75 & 5   & 0.951      & 10            & 0.039              \\
  1    & 5   & 0.852      & 10            & 0.005              \\
  1.25 & 5   & 0.695      & 10            & 0.000              \\
  0.5  & 10  & 1.000      & 20            & 0.427              \\
  0.75 & 10  & 0.998      & 20            & 0.206              \\
  1.0  & 10  & 0.978      & 20            & 0.080              \\
  1.25 & 10  & 0.907      & 20            & 0.025              \\
\end{tabular}
\caption{Static repricing policy fill probabilities}
\label{table:static}
\end{table}

For child-order execution probabilities on the order of $1/2$, we may consider
$y = 0.674$ or $y = 2/3$, which have the following additional probabilities
attached to them.
If we choose $y = 0.674$, the probability of a crossing in the first 10\% of the
window is 0.033, while the probability of an execution during the window is
0.500. The probability of no execution for the parent order is about $0.5^n$,
which for $n = 5$ is about 0.031.
If we change to $y = 2/3$, then the early crossing probability is about 0.035,
the per-child execution probability about 0.505, and the non-execution of any
child orders about 0.033.

\subsubsection{Expected execution price under indefinite static repricing}

Let $y > 0$ be a fixed volatility-scaled price offset for a reprice-until-filled
policy, and let $p = p(y)$ denote the probability that an order is filled. The
number of trials $n$ required to execute is geometrically distributed with
parameter $p$.

Let $z$ denote the expected adverse price movement given no fill. If the first
(and only) fill is at step $k$, then that means that there were $k - 1$ non-fill
steps, with $k$-many adverse price movements of size $z$, followed by a fill at
offset $y$. Hence the expected execution price at step $k$ is $kz + y$ and this
occurs with probability $(1 - p)^{k - 1} p$. Since these events are disjoint, we
obtain the following expression for the expected execution price:
\begin{equation}
    \sum_{j = 0}^\infty (1 - p)^j (j z + y) p = y + \frac{1 - p}{p} z.
\end{equation}
For the standard BM (SBM) price model under consideration, the terms on the
right-hand side balance and cancel, so that the expected execution price is
zero, consistent with the martingale property of SBM.



\subsection{Markov Decision Process formulation}

Now we want to generalize our previous discussion to better suit the practical
purpose. Firstly, we assume a finite time horizon for order execution, which is
usually the case in practice. Secondly, we allow ``repricing" of the order when
submitting new child orders. The need to adjust the price is present in many
practical situations. For example, the trader might place the order more
aggressively when it is near the end of the trading horizon or if there is
strong conviction around the future price movement. 

Suppose we divide the entire trading horizon into multiple time intervals, and
attempt to execute a child order over each interval (stage) until one is filled.
The price we set in each stage affects the likelihood of execution in the
current stage, as well in the future stages, since there will be no future
activities given a fill. These properties can be modeled by a Markov Decision
Process (MDP). 

In short, an MDP provides a mathematical framework for modeling decision-making
in stochastic environments. It is defined by a tuple $(S, C, P, r)$ where $S$ is
the state space, $C$ is the control space, $P(x' | x, c)$ represents the
transition probabilities from state $x$ to $x'$ given control $c$, and $r(x, c)$
is the immediate reward function given state $x$ and control $c$. We specify the
choice of these components mathematically.  

We define the state space $S$ to be the possible number of shares held. In our
simplified problem, this is either 1 share, which we start with, or 0 shares,
following a successful child order execution:
\[
    S = \{1, 0\}.
\]
\begin{remark}
    In practice, any integer value of shares can be divided into sub-problems
    where we hold 1 share. Additionally, fractional shares can be mapped to
    integers.
\end{remark}
The control space $C$ consists of limit orders, determined by
``side",``quantity", and ``limit price". For simplicity, we only consider sell
orders and we assume that an order is always fully filled if executed. In this
case, the control space $C$ is defined by
\[
    C = \{y | y \in \R^+ \}
\]
where $y$ represents the limit price of the sell order. The benchmark is the
asset price at the instance of decision making.

We consider the finite horizon problem with $n$ steps, which can be visualized
in Figure \ref{pic}. There are two states: $1, 0$, where $0$ is an absorbing
state. We start at $t = 0$ in state 1, holding 1 share. At each stage $k$ with
given control $y_k$, the MDP has a chance $P_{y_k}(\Fill)$ of entering state 0,
where it terminates. We enforce the boundary condition that we place a market
order at step $n - 1$ if the order has not been fulfilled. The transition
probability only depends on the current state and control. Hence the problem is
Markov. In fact, the transition probabilities are equal to 
\begin{equation}
    P_{y_k}(\Fill) = P_{y_k}(\LimitFill) + P_{y_k}(\MarketFill)
\end{equation}
and 
\begin{equation}
    P_{y_k}(\NoFill) = 1 - P_{y_k}(\Fill).
\end{equation}
We use \eqref{eq prob limit fill approx} and \eqref{eq prob market fill} to compute the limit fill and market fill probabilities.

\begin{figure}
\centering
\begin{tikzpicture}
    \node (A1) at (0, 1) {1};
    \node (A2) at (2, 1) {1};
    \node (A3) at (4, 1) {1};
    \node (A4) at (6, 1) {1};
    \node (A5) at (7, 1) {...};
    \node (A6) at (8, 1) {1};
    \node (A7) at (10, 1) {1};
    \node (B1) at (2, -0.5) {0};
    \node (B2) at (4, -0.5) {0};
    \node (B3) at (6, -0.5) {0};
    \node (B4) at (7, -0.5) {...};
    \node (B5) at (8, -0.5) {0};
    \node (B6) at (10, -0.5) {0};
    \node (B7) at (12, -0.5) {0};

    \draw[->] (A1) -- (A2);
    \draw[->] (A2) -- (A3);
    \draw[->] (A1) -- (B1);
    \draw[->] (A2) -- (B2);
    \draw[->] (A3)  -- node[sloped, below] {$P_{y_k}(\Fill)$}  (B3);
    \draw[->] (A3)  -- node[above] {$P_{y_k}(\NoFill)$}  (A4);
    \draw[->] (A6) -- (A7);
    \draw[->] (A6) -- (B6);
    \draw[->] (A7) -- (B7);
\end{tikzpicture} 
\caption{MDP transition diagram}
\label{pic}
\end{figure}
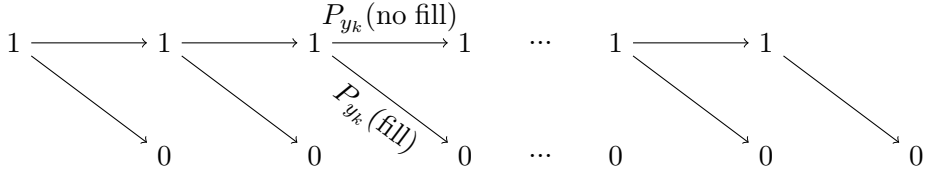

We define the immediate reward function to be the reward we would achieve under
a given state $q$ and control $y$. When we hold no shares, we simply set the
reward to be 0. Otherwise, the reward is stochastic, depending on the price
movement of the Brownian motion during the execution window. Given a limit fill,
the reward is the limit price offset minus applicable transaction costs. Given a
market fill, the reward is the value of the price Brownian motion at the time
when the order reached the order book, minus transaction costs. Given a
non-fill, the reward/loss is the price movement during our inaction. If the
price moves up, then we gain a profit by holding the share; otherwise, we incur
a loss through not having been able to sell. To make this precise, $r(q, y, B)$
is defined by
\begin{equation}
    r(q, y, B) =
    \begin{cases}
        y - \frac{1}{2} s - \makerfee &\text{given a limit fill} \\
        B_\ell - \frac{1}{2} s - \takerfee &\text{given a market fill} \\
        B_1 &\text{given no fill}
    \end{cases}
\end{equation}
Now we define the total reward function. Let the policy be given by $\pi_n =
(y_1, ..., y_n)$ where $y_k$ is the limit price offset at stage $k$. A
straightforward way to measure the profit is by calculating the expected total
reward:
\begin{equation}\label{def expected total reward}
    J_{\pi, n} = \E^\pi[\sum_{k=1}^n r(q_k, y_k, B) ].
\end{equation}
However, \eqref{def expected total reward} does not capture the risk tolerance
of the decision maker. In fact, the longer the order waits for execution, the
greater uncertainty there is in the price. For this reason, although setting
high limit prices in the early stages might result in a higher expected reward,
a risk-averse trader might prefer to sacrifice part of the potential profits in
exchange for smaller variance. In order to measure the level of risk aversion,
we define the following total reward function using exponential utility
function
\begin{equation}\label{def exp total reward}
    J^{\pi, n, \lambda} = - \frac{1}{\lambda} \log \E [e^{-\lambda
    \sum_{k=0}^{n-1} r(q_k, y_k)}].
\end{equation}
We clarify the notation used here:
\begin{enumerate}
    \item $q_k$ is the state at $t = k$, either 0 or 1.
    \item $r(q_k, y_k)$ is the immediate reward random variable given state
    $x_k$ and control $y_k$. By default, $r(0, y_k) = 0$, as the process is
    terminated.
    \item $\E^\pi$ refers to the expectation derived from the probability
    distribution associated with the policy $\pi$, i.e., the transition
    probabilities $p(y_k)$.
\end{enumerate}

We remark that, when $\lambda$ is positive, the exponential utility simulates a
decision maker who is risk-averse. This can be observed from the fact that the
risk index of the exponential utility function is positive. In fact, let 
\[
    U(z) = -\frac{1}{\lambda} e^{-\lambda z}.
\]
The risk index is defined as
\[
    I(z) := - \frac{U''(z)}{U'(z)} = \lambda.
\]
The sign of $I(z)$ indicates the risk tolerance of the investor, see Appendix G of \cite{bertsekas2012dynamic}. We have
\begin{itemize}
    \item $I(z) > 0$: risk-averse, 
    \item $I(z) = 0$: risk neutral,
    \item $I(z) < 0$: risk seeking.
\end{itemize}
Alternatively, one may observe that for small $\lambda$, the total exponential utility is approximately
\[
    J^{\pi, n, \lambda} = \E^\pi [\sum_k r(q_k, y_k)] - \frac{\lambda}{2} \Var[\sum_k r(q_k, y_k)] + \mathcal{O}(\lambda^2).
\]
This approximation can be derived from the second-order Taylor expansion of the
exponential function. Positive $\lambda$ implies that the variance serves as a
penalty.

Now we move on to formulating the backward induction. Denote the value function
as $V_{k}(q)$, where $k$ is the stage number. At stage $n - 1$, we are forced to
place a market order if $x_{n-1} = 1$. Hence, the boundary condition is defined
as
\begin{equation}\label{eq value func t=n-1}
    \begin{aligned}
        & V_{n-1}(1) = - \frac{1}{2} s - \takerfee, \\
        & V_{n-1}(0) = 0.
    \end{aligned}
\end{equation}
The induction relation can be derived as below:
\begin{equation}
    \begin{aligned}\label{eq exp induction}
        V_k(1) & = \sup_{\pi} -\frac{1}{\lambda} \log \E^\pi [e^{-\lambda
        \sum_{j=k}^{n} r_j(x_j, y_j)}] \\
        & = \sup_\pi -\frac{1}{\lambda} \log \left\{\E^{y_k}[e^{-\lambda r_k} |
        \LimitFill] \cdot P_{y_k}(\LimitFill) + \E^{y_k}[e^{-\lambda
        r_k} | \MarketFill] \cdot P_{y_k}(\MarketFill) \right.\\
        & \left.+ \E^\pi[e^{-\lambda r_k} e^{-\lambda \sum_{j = k+1}^{n}} | \NoFill]
        \cdot P_{y_k}(\NoFill) \right\} \\
        & = \sup_{y_k} -\frac{1}{\lambda} \log \left\{ \E^{y_k}[e^{-\lambda r_k} |
        \LimitFill] \cdot P_{y_k}(\LimitFill) + \E^{y_k}[e^{-\lambda r_k} |
        \MarketFill] \cdot P_{y_k}(\MarketFill) \right. \\
        & \left. + \E^{y_k}[e^{-\lambda r_k} | \NoFill] e^{-\lambda V_{k+1}(1)}
        \cdot P_{y_k}(\NoFill) \right\}
    \end{aligned}
\end{equation}
and 
\[
    V_{k}(0) = 0.
\]

We now compute the expectation inside \eqref{eq exp induction}:
\begin{align*}
    I & := \E^{y_k} [e^{-\lambda r(1,
    y_k)} | \Fill] P(\Fill) \\
    II & := \E^{y_k} [e^{-\lambda r(1,
    y_k)} | \NoFill] P(\NoFill) \\
\end{align*}

The first term $I$ is equal to
\begin{align*}
    \E [e^{-\lambda r(1,
    y_k)} | \LimitFill] P(\LimitFill) + \E [e^{-\lambda r(1,
    y_k)} | \MarketFill] P(\MarketFill).
\end{align*}

In the case of a limit fill, the reward $r(x_k, y_k)$ is not stochastic.
However, the reward given a market fill depends on the best available bid at time
$\ell$, and hence is still stochastic:
\begin{align*}
    & r(1, y_k | \LimitFill) = y_k - \frac{1}{2} s - \makerfee, \\
    & r(1, y_k | \MarketFill) = B_\ell - \frac{1}{2} s - \takerfee.
\end{align*}
We compute the expected reward given a market fill by conditioning on $B_\ell$:
\begin{align*}
    \E [e^{-\lambda r(1,
    y_k)} | \MarketFill] 
    & = \frac{1}{P(\MarketFill)} \int_{y_k}^{+\infty} e^{- \lambda (z -
    \frac{1}{2} s - c_{taker})} \frac{1}{\sqrt{l}}f(\frac{z}{\sqrt{l}}) dz \\
    & = \frac{1}{P(\MarketFill)} e^{\lambda (\frac{1}{2}s + c_{taker})}
    \frac{1}{\sqrt{2\pi \ell}} \int_{y_k}^{+\infty} e^{-\lambda z}
    e^{-\frac{z^2}{2\ell}} dz \\
    & = \frac{1}{P(\MarketFill)} e^{\lambda (\frac{1}{2}s + c_{taker})}
    \Phi(-\frac{y_k}{\sqrt{\ell}} - \lambda \sqrt{\ell}) e^{-\frac{\lambda^2
    \ell}{2}}.
\end{align*}
Thus
\begin{equation}
    I = e^{-\lambda (y_k - \frac{1}{2}s - \makerfee)} P(\LimitFill) + e^{\lambda
    (\frac{1}{2}s + \takerfee)} \Phi(-\frac{y_k}{\sqrt{\ell}} - \lambda
    \sqrt{\ell}) e^{-\frac{\lambda^2 \ell}{2}},
\end{equation}
where $l$ is the delay.
The second term $II$ can be evaluated by
\begin{align}\label{eq exp utility II}
    II & = \E[e^{-\lambda (B_1 + V_{k+1}(1))} | \sup_{\ell \leq t < 1} B_t <
    y_k] P(\sup_{\ell \leq t < 1} B_t < y_k) \\
    & = e^{-\lambda V_{k+1}(1)} \int_{-\infty}^{y_k} \int_{-\infty}^{y_k-z}
    e^{-\lambda (z+c)} p_{1-\ell}(c, h \leq y_k - z) dc \frac{1}{\sqrt{\ell}}
    f(\frac{z}{\sqrt{\ell}}) dz,
\end{align}
where $p_{1-l}(c, h \leq y-z)$ is the probability distribution of the close $c$
of a Brownian motion starting from $0$ with volatility $1$ over a time interval
of length $1-\ell$ given its high $h$ satisfies $h < y-z$. Due to scaling, we
have
\[
    p_{1-\ell}(c, h < y-z) = \frac{1}{\sqrt{1-\ell}}p(\frac{c}{\sqrt{1-\ell}}, h
    < \frac{y-z}{\sqrt{1-\ell}}).
\]
We will use an approximation of the double integral in $II$. The derivation of
the approximation can be found in the Appendix.

\subsection{Optimal policy}
We use \eqref{eq exp induction} to compute the optimal policy. In particular, we
set the value function at time $t = n-1$ by \eqref{eq value func t=n-1}. At time
$t = k-1$, we pick the limit price $y_k$ to maximize the RHS of \eqref{eq exp
induction} and recursively perform the backward dynamic programming. Note that,
implicit in the interpretation of the solution is the assumption that we place
no further limit orders once one has been filled.

The key inputs that influence the dynamic programming problem are:
\begin{itemize}
  \item spread
  \item maker/taker fees $\makerfee$, $\takerfee$
  \item maximum number of child orders $n$.
\end{itemize}
In addition, we have the parameter $\lambda > 0$, which indicates risk aversion.
Larger $\lambda$ implies a greater level of risk aversion. It is then expected
that the trader will prefer executing the order as soon as possible in order to
avoid the price uncertainty caused by an extended trading window. This tendency
translates to a lower limit price in the resulting optimal policy. We interpret
it as follows: the risk-averse trader is willing to sacrifice some profit (which
is caused by a lower execution price and a higher chance of paying the taker
fee) in exchange for reduced variance. Below is the plot of the optimal policy
for a range of $\lambda$, where the other parameters are set as: $\takerfee =
0.1, \makerfee = -0.1, \text{spread} = 0.05, \text{delay} = 0.1$.
\begin{figure}[ht!]
    \centering
    \includegraphics[width=\linewidth]{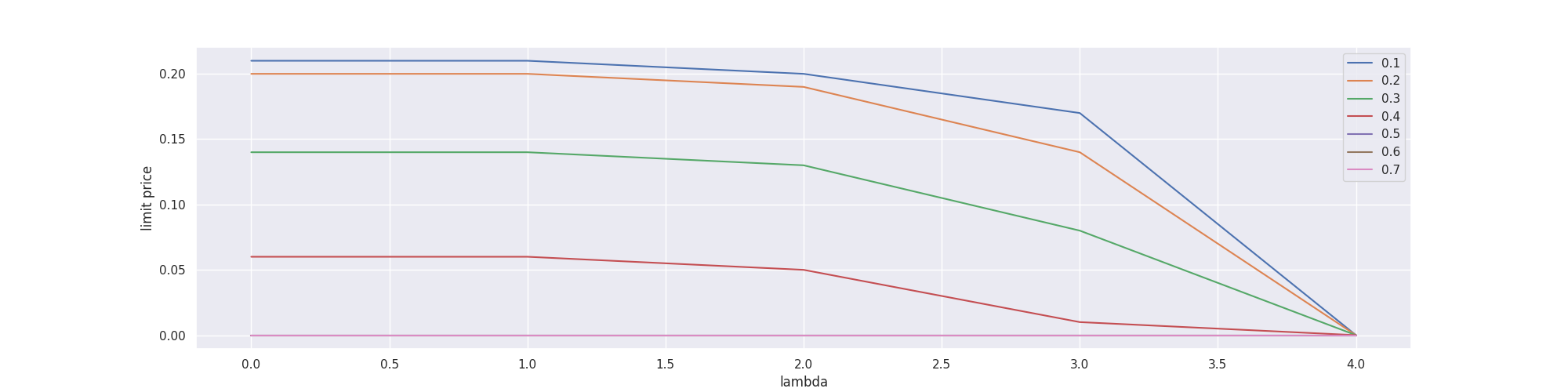}
    \caption{Optimal policy vs $\lambda$}
\end{figure}

\subsection{Discussion}

Interestingly, the relationship of spread to volatility does not appear in the
reduced immediate reward function, and so is not relevant in the MDP. This may
be understood as a consequence of the martingale property of price movements:
while a higher volatility increases the fill probability and expected gain given
a fill, it offsets this expected gain with an expected loss in the case of a
non-fill. On the other hand, the relationship of spread to volatility does not
disappear from the problem entirely. It will influence the particular values of
the value function $v$, which in turn can be compared to the expected instrument
volatility over the course of the parent order. In cases where the spread is
large compared to the expected volatility over the duration of the parent order,
the ``optimal policy" may still be a bad or losing policy. In part, this is by
construction, since we force execution at the boundary if execution is not
achieved earlier.

In practice, we may decide not to trade based upon an unfavorable value function
value. We may also reformulate the problem so that we do not force execution at
the boundary (we only execute passively with limits that are favorable).

Finally, we note that the MDP formulation is not tied to the BM model we have
applied. That is, the probability estimates derived from the model may be
replaced in the MDP formulation with estimates derived from an alternative model
or from empirical statistics if desired (e.g., see \cite{Ba20}[\S 9.1]).

\section{Monte Carlo simulation}
\label{sec: mc sim}
Based on the analysis above, the expected value of a dynamic limit order
repricing policy may be estimated based on various inputs. However, working
strictly within an analytical framework becomes cumbersome as we seek to
understand other properties, such as the standard deviation of the value
function of a policy.

For this reason, we pursue Monte Carlo simulation to estimate other
policy-related statistics of interest.

We first estimate the standard deviation of the policy. Suppose we aim to
execute the parent order over a time interval $I = [0, n]$. In particular, we
split $I$ into sub-intervals $I_k = [k-1, k)$, $k=1,...,n$ and place a limit order at the
start of $I_k$ until the parent order is fully filled. Suppose the limit price
for the $t$-th child order is given by the policy $(y_k)_{k=1,...,n}$, where the
benchmark is the corresponding arrival price at $t=k$. Then the expected
execution price during $I_k$ is equal to the expected arrival price, given that
the child orders were not filled in $I_j$ for $1 \leq j \leq k-1$, plus $y_k$,
the limit price offset. We express it by
\begin{equation}\label{eq expected execution price}
    \sum_{j=1}^{k-1} \E[B_1^{(j)} | \NoFill] + y_k,
\end{equation}
where $B_1^{(j)}$ is the value of the Brownian motion $B_t - B_{j-1}$ at $t =
j$. We use (\ref{eq expected close given no fill}) to compute the expectation.
Below, we show the plot of the expected execution price derived from (\ref{eq
expected execution price}) and compare it with the simulation results.
\begin{figure}[H]
    \centering
    \includegraphics[width=\linewidth]{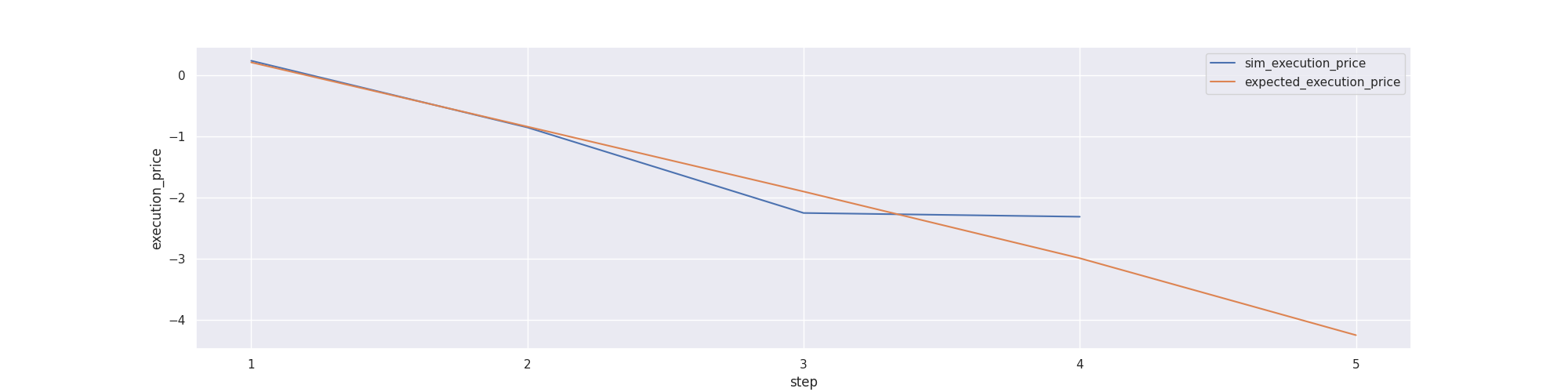}
    \caption{Expected and average simulated execution prices}
\end{figure}

The standard deviation of the execution price scales like
$\mathcal{O}(\sqrt{n})$ where $n$ is the length of the trading horizon. This can
be seen from (\ref{eq expected execution price}) and the square root of time
property of Brownian motion. In the plot below, we compare the rescaled graph of
$y = \sqrt{x}$ and the standard deviation of the simulated net reward with the risk aversion parameter $\lambda = 0.1$.
\begin{figure}[H]
    \centering
    \includegraphics[width=\linewidth]{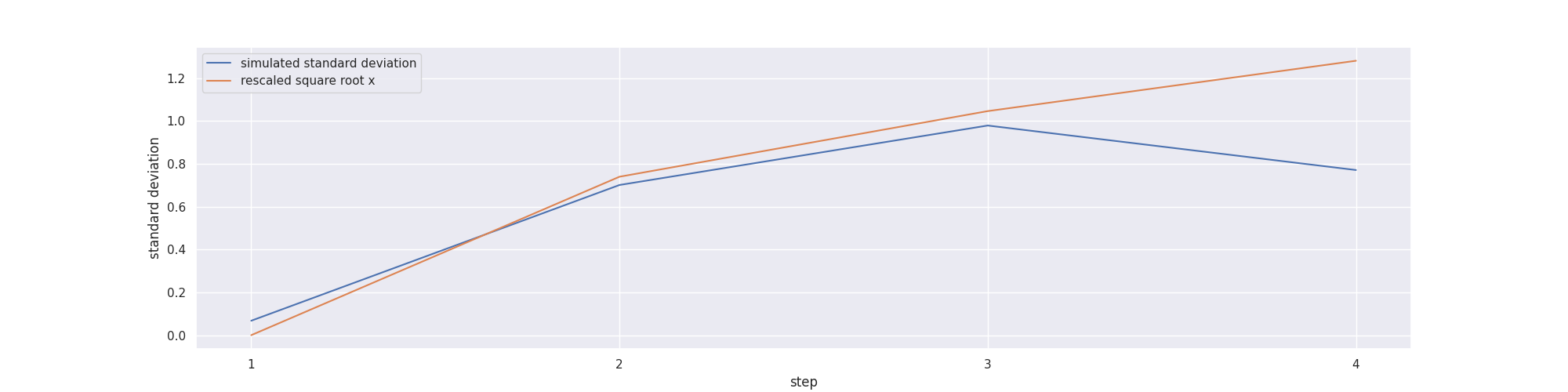}
    \caption{Standard deviation of simulated net reward}
\end{figure}
\begin{remark}
    The last step deviates from $y = c\sqrt{x}$ due to sampling error, as $96\%$
    of the orders in our simulation were executed in step 1 and step 2. 
\end{remark}

Next, we examine how the risk tolerance parameter $\lambda$ is related to the
execution properties, including net profits, variance, percentage of market
order execution, and execution speed. As is expected, the variance decreases as
$\lambda$ increases, because placing the order closer to the market allows for
faster execution, which also reduces price uncertainty as well as net profits.
This can be seen from the following
figures.
\begin{figure}[H]
    \centering
    \includegraphics[width=\linewidth]{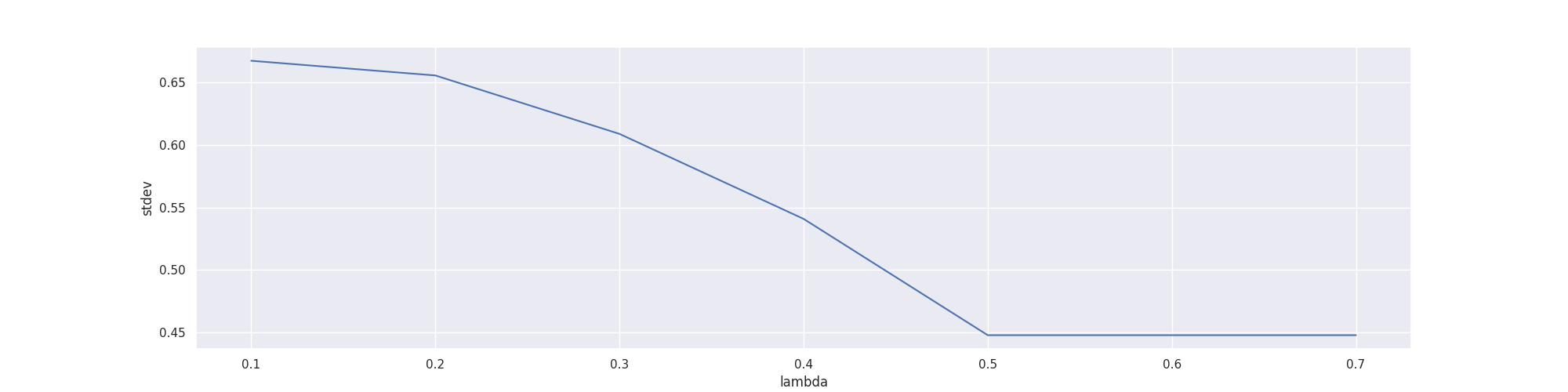}
    \caption{Standard deviation of net reward vs risk aversion parameter $\lambda$}
\end{figure}

The execution speed is represented in the distribution of the step when the
order gets filled. Column values in Table \ref{table:perc} represent the
percentage of filled orders in the specific step.

\begin{table}
\begin{center}
    \begin{tabular}{l|l|l|l|l}
          step  & $\lambda = 0.1$ & $\lambda = 0.3$ & $\lambda = 0.5$ & $\lambda = 0.7$ \\
          1  & 0.809   & 0.845      & 0.920            & 0.920              \\
          2 & 0.157   & 0.133      & 0.077            & 0.077              \\
          3    & 0.028   & 0.019      & 0.003            & 0.003              \\
          4 & 0.006   & 0.003      & 0.000            & 0.000             \\
    \end{tabular}
\end{center}
\caption{Percentage of filled orders in a given step for given risk aversion}
\label{table:perc}
\end{table}
The cost of faster execution and lower variance is decreased net profits. A
greater percentage of orders are executed as market orders, as Figure
\ref{fig:perc} shows.

\begin{figure}[H]
    \centering
    \includegraphics[width=\linewidth]{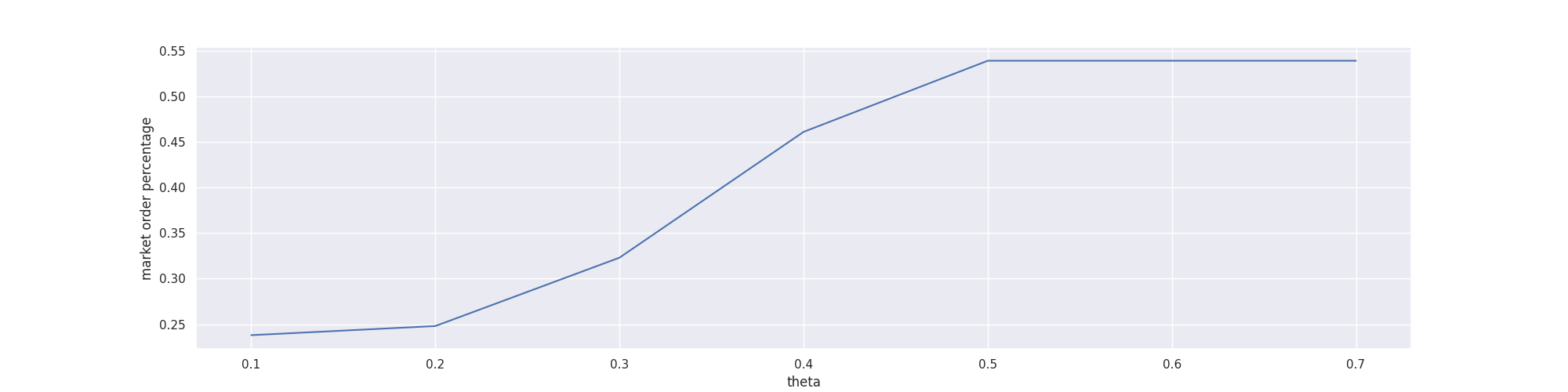}
    \caption{Percentage of market orders}
    \label{fig:perc}
\end{figure}
Of course, the expected execution prices themselves become lower with lower
limit prices. These two factors reduce net profits.
\begin{figure}[H]
    \centering
    \includegraphics[width=\linewidth]{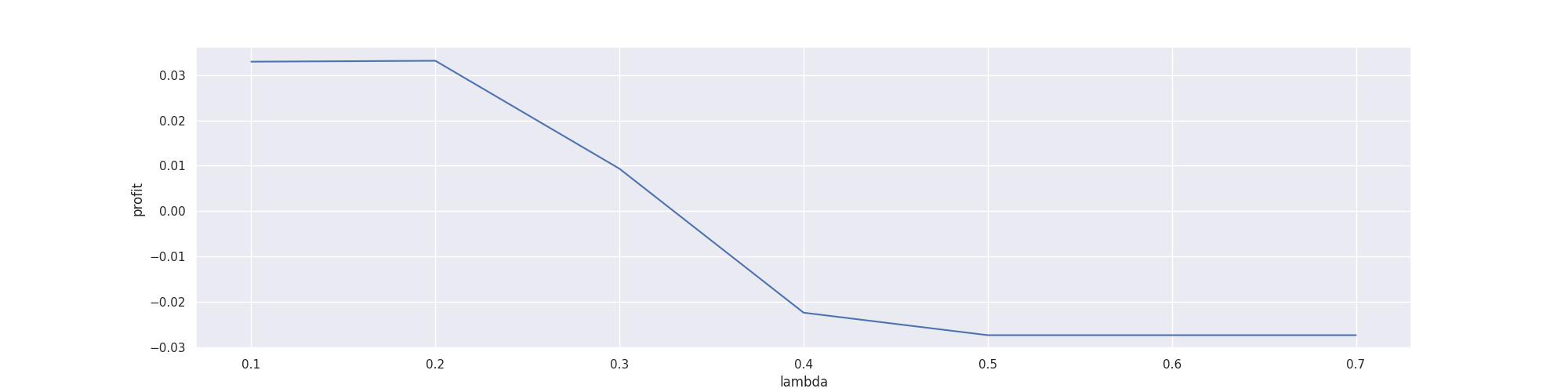}
    \caption{Mean net profits vs $\lambda$}
\end{figure}

\appendix
\section{Calculations}
We begin with properties of Brownian motion related to its high and close. We
let $h$ denote the high, i.e.,
\[
    h = \sup_{0 \leq t \leq 1} B_t,
\]
and $c$ the close
\[
    c = B_1.
\]
As is well known, the joint distribution of the close and high of a standard
Brownian motion on the unit interval is
\[
  p(c, h) = \frac{2(2h - c)}{\sqrt{2\pi}} \exp\left(- \frac{(2h - c)^2}{2} \right)
  \mathbb{1}_{\{h > 0, c < h\}}.
\]
The conditional distribution is, by definition,
\[
    p(c|h) = \frac{p(c, h)}{p(h)}.
\]
A straightforward calculation shows that
\begin{equation} \label{h_pdf}
    p(h) = \int_0^h p(c, h) dc = \sqrt{\frac{2}{\pi}} e^{-\frac{h^2}{2}}
\end{equation}
and so
\begin{equation} \label{c|h_pdf}
    p(c|h) = (2h - c) \exp\left( -\frac{1}{2} \left((2 h - c)^2 - h^2 \right) \right).
\end{equation}



Finally, we compute the marginal probability distribution of the close $c$:
\[
    p(c|h \leq y).
\]
That is, we know that the high $h$ of the SBM did not exceed the limit $y$, but
we do not take into account any other information about $h$ (e.g., its precise
value).

We calculate
\begin{equation} \label{c,h<=y_pdf}
    p(c, h \leq y) =
    \int_0^y p(c, h) dh =
    \frac{1}{\sqrt{2 \pi}} \left(
        e^{-\frac{1}{2}c^2} -
        e^{-\frac{1}{2}(2 y - c)^2}
        \right)
    \mathbb{1}_{\{c < y\}}.
\end{equation}
Combining this with
\[
    p(c|h) p(h) = p(c, h)
\]
and \eqref{h_pdf}, we obtain
\begin{equation}\label{eq prob close given nonfill}
    \begin{aligned}
        p(c|h \leq y) &=
        \frac{\int_0^y p(c|h) p(h) dh}{ \int_0^y p(h) dh} \\
        &= \frac{
            e^{-\frac{1}{2}c^2} -
            e^{-\frac{1}{2}(2 y - c)^2}
        }{
            \sqrt{2 \pi} (1 - 2 \Phi(-y))
        } \mathbb{1}_{\{c < y\}} \\
        &= \frac{1}{\sqrt{2 \pi}} e^{-\frac{1}{2} c^2}
        \frac{1 - e^{-2y(y - c)}}{1 - 2 \Phi(-y)}
        \mathbb{1}_{\{c < y\}}.
    \end{aligned}
\end{equation}

\begin{lem}\label{lem exp inner int}
    We compute $\int_{-\infty}^{y} e^{-\lambda c} p(c, h \leq y) dc$, which is
    used in \eqref{eq exp utility II}. The expression is equal to
\[
    e^{\frac{1}{2}\lambda^2} \Phi(\lambda + y) - e^{\frac{1}{2}\lambda^2 - 2y\lambda} \Phi(\lambda - y).
\]
\end{lem}
\begin{proof}
    By \eqref{c,h<=y_pdf}, we have
    \begin{equation}
        \begin{aligned}
            \int_{-\infty}^y e^{-\lambda c} p(c, h \leq y) dc & = \frac{1}{\sqrt{2\pi}} \int_{-\infty}^y e^{-\frac{1}{2}c^2 - \lambda c} - e^{-\frac{1}{2} c^2 - (\lambda - 2y) c - 2y^2} dc .
        \end{aligned}
    \end{equation}
    Direct computation yields the result.
\end{proof}

\begin{lem}
    Justification of \eqref{eq prob limit fill approx}.
\end{lem}
\begin{proof}
We start with \eqref{eq prob limit fill int}.

\begin{equation}
    \begin{aligned}\label{eq prob limit expansion}
        &\int_{-\infty}^{y} 2\Phi(-\frac{y-z}{\sqrt{1-\ell}}) \frac{1}{\sqrt{\ell}} f(\frac{z}{\sqrt{\ell}}) dz  \\
        = & \int_{-\infty}^{y} 2\Phi(-y) \frac{1}{\sqrt{\ell}} f(\frac{z}{\sqrt{\ell}}) dz + \int_{-\infty}^{y} 2 \left( \Phi(-\frac{y-z}{\sqrt{1-\ell}}) - \Phi(-y) \right) \frac{1}{\sqrt{\ell}} f(\frac{z}{\sqrt{\ell}}) dz \\
        = & 2\Phi(-y) (1 - \Phi(-\frac{y}{\sqrt{\ell}})) 
        + \int_{-\infty}^{y} 2 \left( \Phi(-\frac{y-z}{\sqrt{1-\ell}}) - \Phi(-y) \right) \frac{1}{\sqrt{\ell}} f(\frac{z}{\sqrt{\ell}}) dz \\
    \end{aligned}
\end{equation}

We have for small $y$
\[
    \Phi(y) \approx \frac{1}{2} + \frac{1}{\sqrt{2\pi}} y.
\]
For the integral term in \eqref{eq prob limit expansion}, we can replace
$\Phi(-\frac{y-z}{\sqrt{1-\ell}})$ by $\Phi(-y+z)$ with only lower order errors.
Note that for $y \sim \sqrt{\ell}$,
\begin{enumerate}
    \item The function $\frac{1}{\sqrt{\ell}}f(\frac{z}{\sqrt{\ell}})$ is
    symmetric in $z$.
    \item The function $\Phi(-y + z) - \Phi(-y)$ is ``almost" odd in $z$ for $z
    \in [-y, y]$. So, we can discard the integral from $-y$ to $y$ incurring
    only a lower order error.
\end{enumerate}
This reduces the integral term to
\begin{align*}
    & \int_{-\infty}^{-y} 2\left( \Phi(-y+z) - \Phi(-y) \right) \frac{1}{\sqrt{\ell}} f(\frac{z}{\sqrt{\ell}}) dz \\
    \approx & \frac{2}{\sqrt{2\pi}} \int_{-\infty}^{-y} z \frac{1}{\sqrt{\ell}} f(\frac{z}{\sqrt{\ell}}) dz \\
    = & - \frac{2}{\sqrt{2\pi}} \sqrt{\ell} f(-\frac{y}{\sqrt{\ell}}).
\end{align*}
Collecting all terms yields \eqref{eq prob limit fill approx}.
\end{proof}

The following figure shows that the accuracy of the approximation when $\ell =
0.05$.
\begin{figure}[ht!]
    \centering
    \includegraphics[width=16 cm]{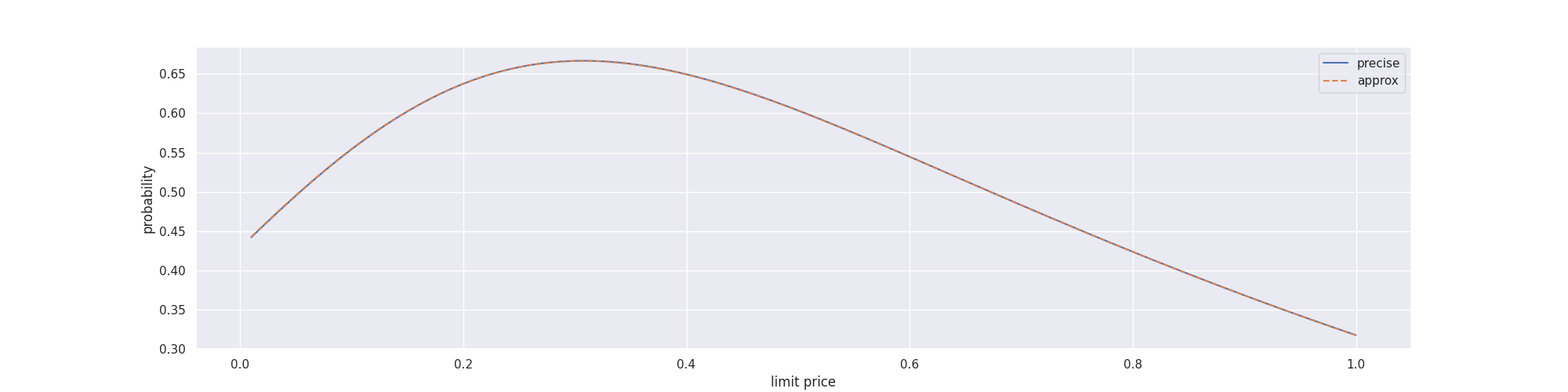}
    \caption{Accuracy of the approximation \eqref{eq prob limit fill approx}}
\end{figure}

\begin{lem}
    We approximate \eqref{eq exp utility II} by
    \begin{align}\label{eq approx exp II}
    & \Phi(\sqrt{1-\ell} \lambda + \frac{y}{\sqrt{1-\ell}}) \Phi(\frac{y}{\sqrt{\ell}} + \lambda \sqrt{\ell}) e^{\frac{1}{2}\lambda^2 \ell} + f(\sqrt{1-\ell}\lambda + \frac{y}{\sqrt{1-\ell}}) \sqrt{\frac{\ell}{1-\ell}} f(\frac{y}{\sqrt{\ell}}) \\
    & - e^{-2\lambda y} \left( 
    \Phi(\sqrt{1-\ell} \lambda - \frac{y}{\sqrt{1-\ell}}) \Phi(\frac{y}{\sqrt{\ell}} - \lambda \sqrt{\ell}) e^{\frac{1}{2}\lambda^2 \ell}
    - f(\sqrt{1-\ell}\lambda - \frac{y}{\sqrt{1-\ell}}) \sqrt{\frac{\ell}{1-\ell}} f(\frac{y}{\sqrt{\ell}})
    \right).
\end{align}
The following figure shows the accuracy of the approximation when $\ell = 0.05, \lambda = 0.3$.
\begin{figure}[ht!]
    \centering
    \includegraphics[width=16 cm]{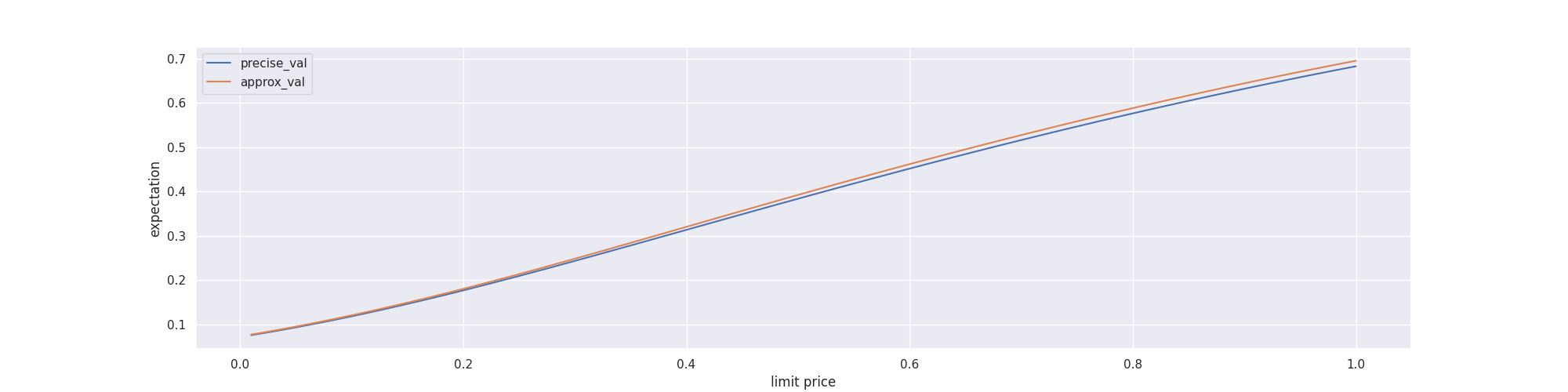}
    \caption{Accuracy of the approximation \eqref{eq approx exp II}}
\end{figure}
\end{lem}

\begin{proof}
    The inner integral is evaluated using lemma \ref{lem exp inner int}:
    \begin{align*}
        & \int_{-\infty}^{y-z} e^{-\lambda (z+c)} p_{1-\ell}(c, h \leq y - z) dc \\
        = & e^{\frac{1}{2} (1-\ell) \lambda^2} e^{-\lambda z} \Phi(\sqrt{1-\ell} \lambda + \frac{y-z}{\sqrt{1-l}}) - e^{-2\lambda(y-z)} \left( \Phi(\sqrt{1-\ell} \lambda - \frac{y-z}{\sqrt{1-\ell}}) \right) \\
        = & e^{\frac{1}{2} (1-\ell) \lambda^2} \left( e^{-\lambda z} \Phi(\sqrt{1-\ell} \lambda + \frac{y-z}{\sqrt{1-\ell}}) - e^{-2\lambda y} e^{\lambda z} \Phi(\sqrt{1-\ell} \lambda - \frac{y-z}{\sqrt{1-\ell}}) \right) \\
        & = e^{\frac{1}{2} (1-\ell) \lambda^2} (I - e^{2\lambda y}\cdot II)
    \end{align*}
    Now we evaluate the outer integral:
    \begin{equation}\label{eq exp outer int}
        \int_{-\infty}^{y} (I - e^{2\lambda y} \cdot II) \frac{1}{\sqrt{\ell}} f(\frac{z}{\sqrt{\ell}}) dz
    \end{equation}
    We ignore the $e^{-\frac{1}{2}(1-\ell)\lambda^2}$ term for now.
    \begin{remark}
        A comparison between the precise and approximate value is carried out at
        the level of \eqref{eq exp outer int}. Namely, we compute \eqref{eq exp
        outer int} numerically, and approximate it below.
    \end{remark}
    We estimate the first term:
    \begin{align*}
        & \int_{-\infty}^{y} I \cdot \frac{1}{\sqrt{\ell}}
        f(\frac{z}{\sqrt{\ell}}) dz \\
        = & \int_{-\infty}^{y} e^{-\lambda z} \Phi(\sqrt{1-\ell}\lambda +
        \frac{y}{\sqrt{1-\ell}}) \frac{1}{\sqrt{\ell}} f(\frac{z}{\sqrt{\ell}})
        dz \\
        + & \int_{-\infty}^{y} e^{-\lambda z} \left( \Phi(\sqrt{1-\ell}\lambda +
        \frac{y-z}{\sqrt{1-\ell}}) - \Phi(\sqrt{1-\ell}\lambda +
        \frac{y}{\sqrt{1-\ell}}) \right) \frac{1}{\sqrt{\ell}}
        f(\frac{z}{\sqrt{\ell}}) dz.
    \end{align*}
    The first term in $\int I$ can be computed precisely:
    \[
        \int I.1 = \Phi(\sqrt{1-\ell} \lambda + \frac{y}{\sqrt{1-\ell}})
        \Phi(\frac{y}{\sqrt{\ell}} + \lambda \sqrt{\ell})
        e^{\frac{1}{2}\lambda^2 \ell}.
    \]
    The second can be estimated by
    \begin{align*}
        \int I.2 & \approx \int_{-\infty}^{y} (1 - \lambda z)
        f(\sqrt{1-\ell}\lambda + \frac{y}{\sqrt{1-\ell}})
        (-\frac{z}{\sqrt{1-\ell}}) \frac{1}{\sqrt{\ell}}
        f(\frac{z}{\sqrt{\ell}}) dz \\
        & \approx f(\sqrt{1-\ell}\lambda + \frac{y}{\sqrt{1-\ell}})
        \sqrt{\frac{\ell}{1-\ell}} f(\frac{y}{\sqrt{\ell}}).
    \end{align*}
    The other term $II$ can be estimated in a similar fashion.
\end{proof}
\bibliographystyle{amsplain}
\bibliography{references}

\end{document}